\documentclass[conference,onecolumn]{IEEEtran}
\IEEEoverridecommandlockouts
\usepackage[utf8]{inputenc}
\usepackage[bookmarks=false]{hyperref}
\usepackage[noadjust]{cite}
\usepackage{soul}
\usepackage{amsmath,amssymb,amsfonts,amsthm,mathtools,bbm}
\usepackage[capitalize]{cleveref}

\usepackage[dvipsnames]{xcolor}
\usepackage{graphicx}
\usepackage{tikz}
\usepackage{balance}

\usepackage[normalem]{ulem}

\newtheorem{theorem}{Theorem}
\newtheorem{definition}{Definition}

\newtheorem{remark}{Remark}
\newtheorem{lemma}{Lemma}

\newcommand{\expectation}{\mathbb{E}}

\usepackage{microtype}
\def\BibTeX{{\rm B\kern-.05em{\sc i\kern-.025em b}\kern-.08em
    T\kern-.1667em\lower.7ex\hbox{E}\kern-.125emX}}
\begin{document}

\title{Joint Identification, Transmission, and Sensing over State-Dependent Channels with Noiseless Feedback
}

\author{
\IEEEauthorblockN{
Vida Gholamian \IEEEauthorrefmark{1}\IEEEauthorrefmark{2},
Marcel A. Mross\IEEEauthorrefmark{1},
Roberto Ferrara\IEEEauthorrefmark{1},
Eduard A. Jorswieck\IEEEauthorrefmark{1},
Christian Deppe\IEEEauthorrefmark{1}\IEEEauthorrefmark{2}
}

\vspace{0.5em}

\IEEEauthorblockA{\IEEEauthorrefmark{1}
Technische Universität Braunschweig, Institute for Communications Technology, Brunswick, Germany}

\IEEEauthorblockA{\IEEEauthorrefmark{2}
BMFTR Research Hub 6G-life, Germany}

\IEEEauthorblockA{
Email: \{v.gholamian-sefiddarboni, m.mross, roberto.ferrara, e.jorswieck, christian.deppe\}@tu-bs.de
}

%\vspace{-0.8cm}
\vspace{-0.4cm}
}

\maketitle

\begin{abstract}
    We consider a discrete memoryless state-dependent channel with noiseless feedback, where the transmitter uses randomized encoding to jointly perform message transmission, identification, and state sensing. In this setting, the transmitter sends a conventional transmission message and an identification message while estimating the unknown channel state sequence under an average distortion constraint. We establish a single-letter characterization of the randomized identification-transmission-sensing capacity-distortion region. The result reveals a single-letter separation structure:
the communication rates are determined by the averaged
channel, subject to the sensing-induced constraint on
admissible input distributions. Specifically, the transmission rate is characterized by mutual information, while the identification rate is characterized by output entropy under the same distortion constraint. The achievability proof combines constant-composition transmission coding, output partitioning, and combinatorial identification coding, while the converse follows from an averaged-channel reduction, distortion lower bounds, and feedback-identification converse arguments.

\end{abstract}

%\begin{IEEEkeywords}
%component, formatting, style, styling, insert
%\end{IEEEkeywords}

\section{Introduction}
Reliable communication over noisy channels traditionally focuses on message
transmission, where one of exponentially many messages is decoded with
vanishing error probability~\cite{shannon1948mathematical}. Identification via
channels, introduced by Ahlswede and Dueck, considers a different task: each
receiver only tests whether a particular message of interest was sent
~\cite{ahlswede1989identification,ja1985identification}. As a result, the
number of identifiable messages can scale as \(2^{2^{nR}}\) with the
blocklength \(n\), for some identification rate \(R>0\)
~\cite{ahlswede1989identification}.

Joint identification and message transmission over DMCs was later studied by
Han and Verd\'u~\cite{han2002new}. For channels with noiseless feedback,
Ahlswede and Dueck showed that feedback can significantly increase
identification capability; in particular, the randomized identification
capacity is characterized by the maximum output entropy rather than the
maximum mutual information~\cite{ahlswedeidentification,egger2024maximal}.
The joint identification-transmission problem with feedback was further
investigated in~\cite{venkatesan2002identification}.

Recently, integrated sensing and communication (ISAC) has attracted growing
attention~\cite{kobayashi2018joint,ahmadipour2025exploring,luo2025integrated}.
Motivated by applications in radar systems, autonomous networks, molecular
communication, and next-generation wireless systems, ISAC seeks to combine
data transmission and environmental sensing using the same resources
~\cite{wu11169758,liu2022information,zheng2026molecular,li2024analysis}.
From an information-theoretic perspective, this has led to communication
models with sensing or state-estimation constraints, often using feedback
observations to estimate unknown channel states~\cite{ahmadipour2022information,kobayashi2018joint,ahmadipour2021joint,chen2023general}.
ISAC has also been studied in the context of identification
~\cite{labidi2024joint,zhao2024achievable}.
However, the joint interplay
between identification, transmission, and sensing over channels with feedback
remains largely unexplored.

In this paper, we study a discrete memoryless state-dependent channel with
noiseless feedback. The state sequence is unknown to both transmitter and
receiver, but its distribution is known. The transmitter uses randomized
encoding to send a transmission message and an identification message while
estimating the state sequence under an average distortion constraint.

Our main contribution is a single-letter characterization of the randomized
identification-transmission-sensing (ITS) capacity-distortion region. We show
that the region exhibits a single-letter separation structure: after averaging
over the i.i.d.\ states, the communication rates are determined by the
averaged DMC with feedback, subject to the sensing-induced constraint
on admissible input distributions. Consequently, for a given transmission rate \(R_{\mathrm T}\),
the optimal identification rate is the largest output entropy achievable under
both the distortion and transmission-rate constraints.

The achievability proof combines constant-composition transmission coding,
equiprobable output partitioning, and combinatorial array constructions
inspired by~\cite{venkatesan2002identification}. The converse follows from
averaged-channel reduction, distortion-constrained input characterization, and
feedback-identification converse arguments. When the distortion constraint is
inactive, the result recovers the classical randomized
identification-transmission region with feedback; otherwise, sensing reduces
the feasible input distributions and can limit both rates.

The remainder of this paper is organized as follows. Section~II introduces the
system model and main result. Section~III provides auxiliary lemmas.
Section~IV presents the proof of the main result, and Section~V concludes
the paper.

\section{System Model and Main Result}

    The state-dependent discrete memoryless channel (DMCS) is
characterized by \((\mathcal X\times\mathcal S,W_S,\mathcal Y)\), where
\(\mathcal X\), \(\mathcal S\), and \(\mathcal Y\) are finite alphabets and $W_S(y|x,s)$
denotes the conditional probability mass function on $\mathcal{Y}$~\cite{csiszar2011information,polyanskiy2025information}.
The state sequence \(S^n\) is i.i.d. according to \(P_S\).
The transmitter and receiver know \(P_S\), but not the realization
of the state sequence. In particular, \(S_t\) is independent of
\((X_t,Y^{t-1})\) under causal encoding.
    Furthermore, define the averaged channel
    \begin{equation}
    \bar W(y|x)
    \coloneqq
    \sum_{s\in\mathcal S}P_S(s)\,W_S(y|x,s).
    \end{equation}
    For an input distribution \(P\in\mathcal P(\mathcal X)\), let $P\bar W$ be the induced
    output distribution defined by
    \begin{equation}
    (P\bar W)(y)
    \coloneqq
    \sum_{x\in\mathcal X} P(x)\bar W(y|x),
    \qquad y\in\mathcal Y.
    \end{equation}
    If \(x^n\in\mathcal X^n\) is transmitted, \(s^n\in\mathcal S^n\) occurs, and
    \(y^n\in\mathcal Y^n\) is received,
    then
    \begin{equation}
    W_S^n(y^n|x^n,s^n)
    =
    \prod_{t=1}^n W_S(y_t|x_t,s_t).
    \end{equation}
Throughout this paper, all logarithms are to base two.
\subsection{Feedback}
    A DMCS $W_S$ with noiseless feedback allows the transmitter
to observe all previous channel outputs without noise.
Thus, before choosing \(X_t\), the encoder can use the previous
channel outputs \(Y^{t-1}\).
An \(n\)-step deterministic feedback strategy is defined as a sequence of functions
\[f=(f_1,\ldots,f_n)\in\mathcal F_n, \qquad
f_t:\mathcal Y^{t-1}\to\mathcal X,\] that defines the encoder via \[X_t=f_t(Y^{t-1})\].
    
    Given a feedback strategy \(f\in\mathcal F_n\), we define
    \begin{align}
    W_S^n(y^n|f,s^n)
    &\coloneqq
    \prod_{t=1}^n
    W_S\bigl(y_t|f_t(y^{t-1}), s_t\bigr).
    \\
    \label{eq:averaged-feedback-channel}
    \bar W^n(y^n|f)
    &\coloneqq
    \prod_{t=1}^n
    \bar W\bigl(y_t|f_t(y^{t-1})\bigr).
    \end{align}
    Notice that by linearity, averaging over the i.i.d. state sequence gives
    \begin{equation}
    \label{eq:state-averaging-identity}
    \sum_{s^n\in\mathcal S^n}
    P_S^n(s^n)W_S^n(y^n|f,s^n)
    =
    \bar W^n(y^n|f).
    \end{equation}
    
    A randomized feedback strategy is a probability distribution \(Q_F\) on
    \(\mathcal F_n\).
    %\vg{should be remove (The transmitter sends a pair \((\alpha,m)\in[N]\times[M]\), where \(\alpha\) is the identification message and \(m\) is the transmission message. For each \((\alpha,m)\), the encoder randomly selects \(f\in\mathcal F^n\) according to \(Q_F(\cdot|\alpha,m)\). The state estimator produces \(\hat S^n=h(X^n,Y^n)\in\mathcal S^n\)).
    %\rf{remove \((\alpha,m)\) from here, the ITS code is a then a family of feedback strategies for each \((\alpha,m)\), but feedback strategies are now defined independently of codes, just like probability distributions and strings exist independently of encoders}}
    
\subsection{Sensing}
    To measure the sensing performance of any possible estimate, we define the per-symbol distortion as
    \begin{equation}
    d_t:=\mathbb{E}\bigl[d(S_t,\hat{S}_t)\bigr],
    \end{equation}
    where \(d:\mathcal S\times\mathcal S\to[0,\infty)\) is an arbitrary
    nonnegative function, which we call the distortion function. The expectation
    is taken with respect to the joint distribution induced by the encoding,
    feedback, channel, and estimation scheme.
    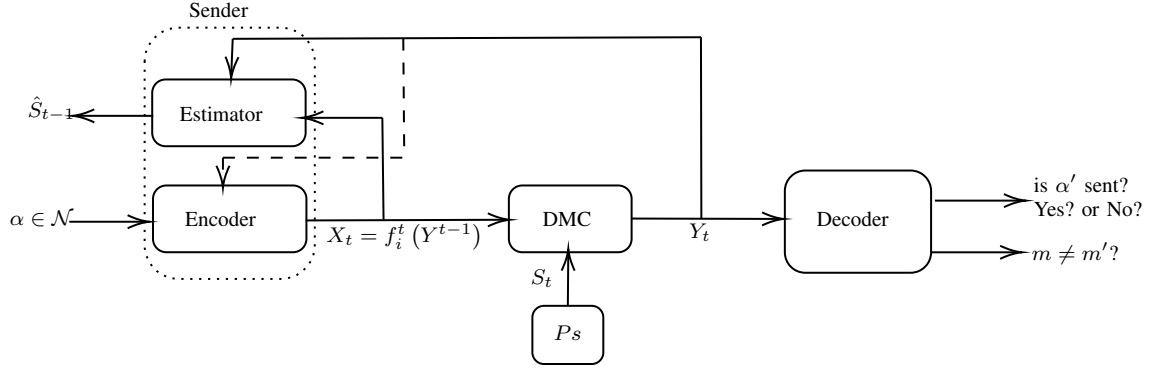
\begin{figure}
        \centering
        \scalebox{1}{% \documentclass{standalone}
% \usepackage{pgfplots}
% \pgfplotsset{compat}
% \usepackage[utf8]{inputenc}

% \begin{document}

\tikzset{every picture/.style={line width=0.75pt}} %set default line width to 0.75pt        

\begin{tikzpicture}[x=0.75pt,y=0.75pt,yscale=-1,xscale=1]

%Rounded Rect [id:dp3889560822286926] 
\draw   (82,61.7) .. controls (82,58) and (85,55) .. (88.7,55) -- (152.3,55) .. controls (156,55) and (159,58) .. (159,61.7) -- (159,81.8) .. controls (159,85.5) and (156,88.5) .. (152.3,88.5) -- (88.7,88.5) .. controls (85,88.5) and (82,85.5) .. (82,81.8) -- cycle ;
%Rounded Rect [id:dp7317105833299926] 
\draw   (82.4,114.9) .. controls (82.4,111.2) and (85.4,108.2) .. (89.1,108.2) -- (152.7,108.2) .. controls (156.4,108.2) and (159.4,111.2) .. (159.4,114.9) -- (159.4,135) .. controls (159.4,138.7) and (156.4,141.7) .. (152.7,141.7) -- (89.1,141.7) .. controls (85.4,141.7) and (82.4,138.7) .. (82.4,135) -- cycle ;
%Rounded Rect [id:dp5869180715564811] 
\draw   (261,114.9) .. controls (261,111.2) and (264,108.2) .. (267.7,108.2) -- (315.9,108.2) .. controls (319.6,108.2) and (322.6,111.2) .. (322.6,114.9) -- (322.6,135) .. controls (322.6,138.7) and (319.6,141.7) .. (315.9,141.7) -- (267.7,141.7) .. controls (264,141.7) and (261,138.7) .. (261,135) -- cycle ;
%Rounded Rect [id:dp573714317425606] 
\draw   (399.5,110.96) .. controls (399.5,105.29) and (404.09,100.7) .. (409.76,100.7) -- (462.64,100.7) .. controls (468.31,100.7) and (472.9,105.29) .. (472.9,110.96) -- (472.9,141.74) .. controls (472.9,147.41) and (468.31,152) .. (462.64,152) -- (409.76,152) .. controls (404.09,152) and (399.5,147.41) .. (399.5,141.74) -- cycle ;
%Straight Lines [id:da26822684871036306] 
\draw    (159,125.75) -- (259.4,125.41) ;
\draw [shift={(261.4,125.4)}, rotate = 179.8] [color={rgb, 255:red, 0; green, 0; blue, 0 }  ][line width=0.75]    (10.93,-3.29) .. controls (6.95,-1.4) and (3.31,-0.3) .. (0,0) .. controls (3.31,0.3) and (6.95,1.4) .. (10.93,3.29)   ;
%Straight Lines [id:da3243952830561859] 
\draw    (322.4,124.8) -- (397,124.99) ;
\draw [shift={(399,125)}, rotate = 180.15] [color={rgb, 255:red, 0; green, 0; blue, 0 }  ][line width=0.75]    (10.93,-3.29) .. controls (6.95,-1.4) and (3.31,-0.3) .. (0,0) .. controls (3.31,0.3) and (6.95,1.4) .. (10.93,3.29)   ;
%Straight Lines [id:da5317290285219953] 
\draw    (357.6,125) -- (357.4,33.8) ;
%Straight Lines [id:da7370686023259566] 
\draw    (357.4,33.8) -- (122.4,34.6) ;
%Rounded Rect [id:dp9523057276487921] 
\draw  [dash pattern={on 0.84pt off 2.51pt}] (78.2,46.08) .. controls (78.2,36.65) and (85.85,29) .. (95.28,29) -- (146.52,29) .. controls (155.95,29) and (163.6,36.65) .. (163.6,46.08) -- (163.6,138.02) .. controls (163.6,147.45) and (155.95,155.1) .. (146.52,155.1) -- (95.28,155.1) .. controls (85.85,155.1) and (78.2,147.45) .. (78.2,138.02) -- cycle ;
%Straight Lines [id:da9324466775894953] 
\draw    (122.4,34.6) -- (121.68,52.6) ;
\draw [shift={(121.6,54.6)}, rotate = 272.29] [color={rgb, 255:red, 0; green, 0; blue, 0 }  ][line width=0.75]    (10.93,-3.29) .. controls (6.95,-1.4) and (3.31,-0.3) .. (0,0) .. controls (3.31,0.3) and (6.95,1.4) .. (10.93,3.29)   ;
%Rounded Rect [id:dp5361162875307939] 
\draw   (272.8,174.44) .. controls (272.8,171.21) and (275.41,168.6) .. (278.64,168.6) -- (302.56,168.6) .. controls (305.79,168.6) and (308.4,171.21) .. (308.4,174.44) -- (308.4,191.96) .. controls (308.4,195.19) and (305.79,197.8) .. (302.56,197.8) -- (278.64,197.8) .. controls (275.41,197.8) and (272.8,195.19) .. (272.8,191.96) -- cycle ;
%Straight Lines [id:da8423751886093759] 
\draw    (290.6,168.4) -- (290.82,142) ;
\draw [shift={(290.83,140)}, rotate = 90.47] [color={rgb, 255:red, 0; green, 0; blue, 0 }  ][line width=0.75]    (10.93,-3.29) .. controls (6.95,-1.4) and (3.31,-0.3) .. (0,0) .. controls (3.31,0.3) and (6.95,1.4) .. (10.93,3.29)   ;
%Straight Lines [id:da7739815274076867] 
\draw    (474.8,115.8) -- (518.4,115.8) ;
\draw [shift={(520.4,115.8)}, rotate = 180] [color={rgb, 255:red, 0; green, 0; blue, 0 }  ][line width=0.75]    (10.93,-3.29) .. controls (6.95,-1.4) and (3.31,-0.3) .. (0,0) .. controls (3.31,0.3) and (6.95,1.4) .. (10.93,3.29)   ;
%Straight Lines [id:da013839737422059906] 
\draw    (197.5,74.25) -- (159,74.25) ;
\draw [shift={(157,74.25)}, rotate = 360] [color={rgb, 255:red, 0; green, 0; blue, 0 }  ][line width=0.75]    (10.93,-3.29) .. controls (6.95,-1.4) and (3.31,-0.3) .. (0,0) .. controls (3.31,0.3) and (6.95,1.4) .. (10.93,3.29)   ;
%Straight Lines [id:da4936010085189688] 
\draw    (197.5,74.25) -- (198.1,125.3) ;
%Straight Lines [id:da2955781193829474] 
\draw  [dash pattern={on 4.5pt off 4.5pt}]  (117.5,94.25) -- (117.5,107.25) ;
\draw [shift={(117.5,109.25)}, rotate = 270] [color={rgb, 255:red, 0; green, 0; blue, 0 }  ][line width=0.75]    (10.93,-3.29) .. controls (6.95,-1.4) and (3.31,-0.3) .. (0,0) .. controls (3.31,0.3) and (6.95,1.4) .. (10.93,3.29)   ;
%Straight Lines [id:da9424487406917154] 
\draw  [dash pattern={on 4.5pt off 4.5pt}]  (117.5,94.25) -- (208.5,94.25) ;
%Straight Lines [id:da6052237273821827] 
\draw  [dash pattern={on 4.5pt off 4.5pt}]  (208,34.5) -- (208.5,94.25) ;
%Straight Lines [id:da09657415978875161] 
\draw    (82.5,73.25) -- (44,73.25) ;
\draw [shift={(42,73.25)}, rotate = 360] [color={rgb, 255:red, 0; green, 0; blue, 0 }  ][line width=0.75]    (10.93,-3.29) .. controls (6.95,-1.4) and (3.31,-0.3) .. (0,0) .. controls (3.31,0.3) and (6.95,1.4) .. (10.93,3.29)   ;
%Straight Lines [id:da20965367670816082] 
\draw    (40.4,126.4) -- (78.6,126.4) ;
\draw [shift={(80.6,126.4)}, rotate = 180] [color={rgb, 255:red, 0; green, 0; blue, 0 }  ][line width=0.75]    (10.93,-3.29) .. controls (6.95,-1.4) and (3.31,-0.3) .. (0,0) .. controls (3.31,0.3) and (6.95,1.4) .. (10.93,3.29)   ;
%Straight Lines [id:da44419062508955887] 
\draw    (472.9,141.74) -- (516.5,141.74) ;
\draw [shift={(518.5,141.74)}, rotate = 180] [color={rgb, 255:red, 0; green, 0; blue, 0 }  ][line width=0.75]    (10.93,-3.29) .. controls (6.95,-1.4) and (3.31,-0.3) .. (0,0) .. controls (3.31,0.3) and (6.95,1.4) .. (10.93,3.29)   ;

% Text Node
\draw (94.4,67) node [anchor=north west][inner sep=0.75pt]  [font=\footnotesize] [align=left] {Estimator};
% Text Node
\draw (97,120) node [anchor=north west][inner sep=0.75pt]  [font=\footnotesize] [align=left] {Encoder};
% Text Node
\draw (276.6,120) node [anchor=north west][inner sep=0.75pt]  [font=\footnotesize] [align=left] {DMC};
% Text Node
\draw (414,120) node [anchor=north west][inner sep=0.75pt]  [font=\footnotesize] [align=left] {Decoder};
% Text Node
\draw (282,177) node [anchor=north west][inner sep=0.75pt]  [font=\footnotesize]  {$Ps$};
% Text Node
\draw (168,126) node [anchor=north west][inner sep=0.75pt]  [font=\footnotesize]  {$X_{t} =f_{i}^{t}\left( Y^{t-1}\right)$};

% Text Node
\draw (270.56,148.8) node [anchor=north west][inner sep=0.75pt]  [font=\footnotesize]  {$S_{t}$};
% Text Node
\draw (18,62) node [anchor=north west][inner sep=0.75pt]  [font=\footnotesize]  {$\hat{S}_{t-1}$};
% Text Node
\draw (350,126) node [anchor=north west][inner sep=0.75pt]  [font=\footnotesize]  {$Y_{t}$};
% Text Node
\draw (99,15) node [anchor=north west][inner sep=0.75pt]  [font=\footnotesize] [align=left] {Sender};
% Text Node
\draw (523,101) node [anchor=north west][inner sep=0.75pt]  [font=\footnotesize] [align=left] {is $\displaystyle \alpha '$ sent?\\Yes? or No?};
% Text Node
\draw (9,119) node [anchor=north west][inner sep=0.75pt]  [font=\footnotesize] [align=left] {$\displaystyle \alpha\in \mathcal{N}$};
% Text Node
\draw (522,135) node [anchor=north west][inner sep=0.75pt]  [font=\footnotesize] [align=left] {$\displaystyle m \neq m'$?};

\end{tikzpicture}

% \end{document}
}
        \caption{State-dependent channel with noiseless feedback.}
        \label{fig:System}
    \end{figure}
    The average block distortion for message pair \((\alpha,m)\) is defined as
\begin{equation}
\bar d^{(n)}_{\alpha,m}
:=
\frac{1}{n}\sum_{t=1}^n
\mathbb E\bigl[d(S_t,\hat S_t)\mid
(\alpha,m)\text{ is sent}\bigr].
\end{equation}

    We define an optimal estimator as follows.
\begin{equation}
\label{eq:hs}
h^\star(x,y)
\in
\arg\min_{\hat s\in{\mathcal S}}
\sum_{s\in\mathcal S}
P_{S|X,Y}(s|x,y)d(s,\hat s).
\end{equation}
If several minimizers exist, one of them is selected arbitrarily.
Here,
\[
P_{S|X,Y}(s|x,y)
=
\frac{P_S(s)W_S(y|x,s)}{\bar W(y|x)}
\]
whenever \(\bar W(y|x)>0\).
    
    We define the minimum distortion achieved by the optimal estimator associated with each input symbol
    \(x\in\mathcal X\) as
    \begin{equation}
    \label{eq:dstar}
    d^\star(x)
    \coloneqq
    \mathbb{E}\bigl[d(S,h^\star(x,Y))\mid X=x\bigr],
    \end{equation}
    and the corresponding minimum distortion associated with an input distribution
    \(P\in\mathcal P(\mathcal X)\) as
    \begin{equation}
    d^\star(P)
    \coloneqq \expectation[d^\star(X)]=
    \sum_{x\in\mathcal X}P(x)\,d^\star(x). \end{equation}

\subsection{Identification, Transmission, and Sensing}
  At the receiver, two tasks are performed jointly:
for a queried identity \(\alpha'\), it decides whether
\(\alpha'\) matches the sent identification message \(\alpha\)
and, if so, decodes the transmission message \(m\).

    \begin{definition}    
    An \((n,N,M,\lambda_1,\lambda_2,D)\) randomized ITS code for the
    state-dependent channel \(W_S\) consists of
        \[\left(
        \left\{
        \bigl(Q_F(\cdot|\alpha,m),\mathcal{D}_{\alpha,m}\bigr)
        :
        (\alpha,m)\in [N]\times[M]
        \right\},h
        \right),\]
        where \(Q_F(\cdot|\alpha,m)\in\mathcal{P}(\mathcal{F}_n)\) is an \(n\)-step
        randomized feedback strategy associated with the message pair
        \((\alpha,m)\), \(\mathcal{D}_{\alpha,m}\subseteq\mathcal{Y}^n\) is the
        corresponding decoding set satisfying
        \begin{equation}
        \mathcal{D}_{\alpha,m}\cap\mathcal{D}_{\alpha,m'}=\emptyset
        \end{equation}
        for all \(\alpha\in[N]\) and all \(m\neq m'\),
        and
        $h:\mathcal X^n\times\mathcal Y^n\to\mathcal S^n$
        is the state estimator. For every \(\alpha\in[N]\), define
        \begin{equation}
        \label{eq:Da}
        \mathcal{D}_{\alpha}
        \coloneqq
        \bigcup_{m\in[M]}\mathcal{D}_{\alpha,m}.
        \end{equation}
        The joint transmission and missed identification (type 1)
error probability satisfies
        \begin{align}
        % &\frac{1}{M}\sum_{m=1}^M
        % \sum_{s^n\in\mathcal{S}^n} P_S^n(s^n)
        % \sum_{f\in\mathcal{F}_n} Q_F(f|\alpha,m)
        % W_S^n\bigl(
        % \mathcal{D}_{\alpha,m}^c \,\big|\, f,s^n
        % \bigr)
        % \le \lambda_1
        &\frac{1}{M}\sum_{m=1}^M
        \sum_{f\in\mathcal{F}_n} Q_F(f|\alpha,m)
        \bar W ^n\bigl(
        \mathcal{D}_{\alpha,m}^c \,\big|\, f
        \bigr)
        \le \lambda_1
        \label{eq:its_cond1}
        \intertext{for all $\alpha\in[N]$, the false identification (type 2) error probability satisfies}
        % &\frac{1}{M}\sum_{m=1}^M
        % \sum_{s^n\in\mathcal{S}^n} P_S^n(s^n)
        % \sum_{f\in\mathcal{F}_n} Q_F(f|\alpha,m)
        % W_S^n\bigl(
        % \mathcal{D}_{\alpha'} \,\big|\, f,s^n
        % \bigr) \le \lambda_2
        &\frac{1}{M}\sum_{m=1}^M
        \sum_{f\in\mathcal{F}_n} Q_F(f|\alpha,m)
        \bar W ^n\bigl(
        \mathcal{D}_{\alpha'} \,\big|\, f
        \bigr) \le \lambda_2
        \label{eq:its_cond2}
        \end{align}
        for all $\alpha\neq\alpha'\in[N]$, and the maximal distortion of the code satisfies
        \begin{equation}
        \bar d^{(n)} \coloneqq
        \max_{(\alpha,m)\in[N]\times[M]}
        \bar d^{(n)}_{\alpha,m} \le D.
        \end{equation}
    \end{definition}

  \begin{definition}
A rate pair $(R_{\mathrm{ID}},R_{\mathrm T})$ is said to be
achievable under distortion constraint \(D\) for the
state-dependent channel \(W_S\) with randomized encoding
if there exists a sequence of
\((n,N_n,M_n,\lambda_{1,n},\lambda_{2,n},D_n)\)
randomized ITS codes such that
\begin{equation}
\lim_{n\to\infty}\lambda_{1,n}
=
\lim_{n\to\infty}\lambda_{2,n}=0,
\qquad
\limsup_{n\to\infty}D_n\le D,
\end{equation}
and
\begin{equation}
\liminf_{n\to\infty}\frac{1}{n}\log\log N_n
\ge R_{\mathrm{ID}},
\qquad
\liminf_{n\to\infty}\frac{1}{n}\log M_n
\ge R_{\mathrm T}.
\end{equation}
\end{definition}

    \begin{definition}%[Capacity--distortion region]
    The randomized ITS capacity-distortion region,
    denoted by
    $\mathcal{C}_{\mathrm{ID,T}}^{\mathrm r}(D)$,
    is defined as the set of all achievable rate pairs
    $(R_{\mathrm{ID}},R_{\mathrm T})$ under distortion constraint \(D\).
    
    For a fixed transmission rate \(R_{\mathrm T}\), the randomized
    ITS capacity-distortion function is defined as
    \begin{equation}
        C^r_{\mathrm{ID},\mathrm T}(R_{\mathrm T},D)
        :=
        \sup\left\{
        R_{\mathrm{ID}}:
        (R_{\mathrm {ID}},R_{\mathrm{T}})
        \in
        \mathcal C^r_{\mathrm{ID},\mathrm T}(D)
        \right\}.
    \end{equation}
    \end{definition}
    %We now characterize the randomized ITS capacity distortion tradeoff.
    We now characterize the ITS tradeoff.
    \begin{theorem}\label{thm:main}
  Assume that $C(\bar W):=
\max_{P\in\mathcal P(\mathcal X)}I(P;\bar W)>0$.
For every $D\ge D_{\min}:=\min_{x\in\mathcal X}d^\star(x)$,
the randomized ITS capacity-distortion region
of the state-dependent channel $W_S$ is given by
    \begin{equation}
    \begin{aligned}\label{eq:CDregion}
    \mathcal{C}_{\mathrm{ID,T}}^r(D)
    &=
    \bigcup_{P\in\mathcal P_{D}}
    \!\left\{
    (R_{\mathrm{ID}},R_{\mathrm T}):\right.
    \left.
   0\le R_{\mathrm T}\le I(P;\bar W),\,
0\le R_{\mathrm{ID}}\le H(P\bar W)
    \right\}.
    \end{aligned}
    \end{equation}
    where
    $\mathcal{P}_D
    =
    \left\{
    P\in\mathcal{P}(\mathcal{X}):
    d^\star(P)\le D
    \right\}$.
    
    Equivalently, for every
\(0\le R_{\mathrm T}\le \max_{P\in\mathcal P_D}I(P;\bar W)\),
    \[
    C_{\mathrm{ID,T}}^{\mathrm r}(R_{\mathrm T},D)
    =
    \max_{\substack{
    P\in\mathcal{P}_D:\\
    I(P;\bar W)\ge R_{\mathrm T}
    }}
    H(P\bar W).
    \]
    
    \end{theorem}
    
    \begin{remark}
  The randomized ITS capacity-distortion region exhibits a
single-letter separation structure: the communication rates
are determined by the averaged channel $\bar W$, subject to
the sensing-induced input constraint $P\in\mathcal P_D$.
    If the distortion is unconstrained, i.e., $\mathcal P_{D}=\mathcal P(\mathcal X)$,
    the region reduces to the classical randomized identification-transmission
    region with feedback~\cite{venkatesan2002identification} for $\bar W$.
    \end{remark}

\section{Auxiliary Lemmas}

    In this section, we collect several auxiliary lemmas that will be used in the proof of the main theorem.

    The following lemma shows that the symbolwise estimator \(h^\star\) provides a lower bound on the distortion achieved by any estimator.
    \begin{lemma}
    \label{lem:dist_lb}
    Consider any randomized $(n,N,M,\lambda_1,\lambda_2,D)$ ITS code and fix
    $(\alpha,m)\in[N]\times[M]$. Let
    \begin{equation}
    \begin{aligned}
    \label{eq:Pbar}
    P_{X_t|\alpha,m}(x)
    &:= \Pr\{X_t=x \mid \alpha,m\}, \\
    \bar P_{\alpha,m}^{(n)}
    &:= \frac{1}{n}\sum_{t=1}^n P_{X_t|\alpha,m}.
    \end{aligned}
    \end{equation}
    Then
    \begin{equation}
    \bar d_{\alpha,m}^{(n)}
    \ge
    d^\star\bigl(\bar P_{\alpha,m}^{(n)}\bigr).
    \end{equation}
    \end{lemma}
   \begin{proof}
Fix $(\alpha,m)$ and $t\in[n]$.
Since the states are i.i.d. and the encoder uses causal feedback,
the posterior distribution of $S_t$ given $(X^n,Y^n)$
depends only on $(X_t,Y_t)$.
Thus, by the optimality of $h^\star$,
    \[
    \begin{aligned}
    \mathbb{E}&[d(S_t,\hat S_t)|\alpha,m]
    \ge
    \mathbb{E}[d(S_t,h^\star(X_t,Y_t))|\alpha,m]  \\
    &=
    \sum_{x\in\mathcal X} P_{X_t|\alpha,m}(x)
    \mathbb{E}[d(S_t,h^\star(x,Y_t))|X_t=x,\alpha,m] 
    =
    \sum_{x\in\mathcal X} P_{X_t|\alpha,m}(x)d^\star(x)
    =
    d^\star(P_{X_t|\alpha,m}).
    \end{aligned}
    \]
    %Averaging over $t=1,\ldots,n$ and using the linearity of $d^\star(P)$ in $P$ gives
    Averaging over \(t\) and using the linearity of \(d^\star(P)\) in
\(P\in\mathcal P(\mathcal X)\) gives
    %{\small
    \begin{align}
    \bar d_{\alpha,m}^{(n)} \ge
    \frac{1}{n}\sum_{t=1}^n d^\star(P_{X_t|\alpha,m}) \notag
    =
    d^\star\!\left(\frac{1}{n}\sum_{t=1}^n P_{X_t|\alpha,m}\right)
    =
    d^\star(\bar P_{\alpha,m}^{(n)}).
    \end{align}
    %}
    \end{proof}
    %\subsection{Equitype Transmission Code with Partition Property}
    
    The following lemma provides the constant-composition transmission code
    construction used in the achievability proof for the averaged channel
    \(\bar W\).
    
    \begin{lemma}[{\cite[Lemma~3.1]{venkatesan2002identification}, adapted to \(\bar W\)}]%[Equitype transmission code with partition property]
    \label{lem:equitype_partition}
    Fix an input distribution
    $P\in\mathcal{P}(\mathcal{X})$
    and let $P_n$ be a sequence of $n$-types such that
    $P_n\to P$.
    Let
    \[
    R' < I(P;\bar{W}),
    \qquad
    R'' < H(\bar{W}\mid P).
    \]
    Then, for all sufficiently large $n$, there exist:
    
    \begin{itemize}
        \item a family of codewords $c_{m,\ell}\in\mathcal{X}^n$, $(m,\ell)\in[M_n]\times[L_n]$, all of type $P_n$, satisfying \[\frac{1}{n}\log(M_nL_n)\le R',\]
    
        \item a partition \(\{C_{m,\ell}\}\) of \(\mathcal Y^n\) such that
        \(\bar W^n(C_{m,\ell}^{c}\mid       c_{m,\ell})\le \alpha_n\) for all
        \((m,\ell)\), with \(\alpha_n\to0\) exponentially fast,
    
       \item for every $(m,\ell)$, a further partition
        $C_{m,\ell}=C_{m,\ell,e}\,\dot\cup\,
        \bigl(\dot\bigcup_{k=1}^{K_n}C_{m,\ell,k}\bigr)$,
        for any positive integer $K_n$ satisfying
        \[\frac1n\log K_n\le R'',\qquad
        \bar W^n(C_{m,\ell,e}\mid c_{m,\ell})\le\beta_n,\]
        and
        \[
        \sum_{k=1}^{K_n}
        \left|
        \bar W^n(C_{m,\ell,k}\mid c_{m,\ell})-\frac1{K_n}
        \right|\le\delta_n,
        \]
        with $\beta_n,\delta_n\to0$ exponentially fast.
        \end{itemize}
        \end{lemma}
    
    %\subsection{Array Lemma}
    
    The following combinatorial lemma is used to construct the identification
    component of the coding scheme.
    
  \begin{lemma}[{\cite[Lemma~3.2]{venkatesan2002identification}}]
\label{lem:array}
Fix $\varepsilon\in(0,1)$ and an integer $J\ge 2/\varepsilon$. Then there
exists $c_{\varepsilon}>0$ such that, if
$\log N_n \le c_{\varepsilon} S_n$ for a positive integer $S_n$, then, for
all sufficiently large $n$, there exists an array
$A:[N_n]\times[S_n]\to[J]$,
whose any two distinct rows have Hamming distance larger than
$(1-\varepsilon)S_n$.
\end{lemma}

The approximate equiprobability follows by restricting
each decoding region to conditionally typical output
sequences and partitioning their probability masses
into $K_n$ nearly equal parts. Since
$R''<H(\bar W|P)$, the maximum probability of
an individual retained output sequence is
exponentially smaller than $1/K_n$.
Consequently, the total deviation from the uniform
distribution vanishes exponentially.
    
    In the achievability proof, \Cref{lem:equitype_partition,lem:array} are used with
    %\begin{equation}
     $ S_n=M_nL_nK_n.$  
   % \end{equation}
    %The following lemma states the converse bound for randomized
    %identification transmission over a discrete memoryless channel with
    %noiseless feedback and serves as the main ingredient of the converse proof.
    %\begin{lemma}%[Converse for randomized identification--transmission over a DMC]
    %\label{lem:IT_converse}
    %Consider a DMC \(W:\mathcal X\to\mathcal Y\) with noiseless feedback.
    %For every sequence of randomized identification--transmission codes with
    %blocklength \(n\), transmission size \(M\), identification size \(N\), and
    %vanishing error probabilities, there exists an input distribution
    %\(Q_n\in\mathcal P(\mathcal X)\) such that
    %$\frac1n\log M
    %\le
    %I(Q_n;W)+\delta_n$,
    %and
    %$\frac1n\log\log N
    %\le
    %H(Q_nW)+\delta_n$,
    %where \(\delta_n\to0\).
    %\end{lemma}
    
    The next lemma is a distortion-constrained version of the trimming and
    counting argument for identification with feedback, following
    Venkatesan and Anantharam~\cite{venkatesan2002identification}.
    \begin{lemma}%[\cite{venkatesan2002identification}]
    \label{lem:dist_trim_count}
    Consider the averaged DMC \(\bar W:\mathcal X\to\mathcal Y\) with noiseless
    feedback. Fix \(\alpha\in[N]\). For each \(m\in[M]\), let
    \[Q_F(\cdot|\alpha,m)\in\mathcal P(\mathcal F_n)\] be the randomized feedback
    strategy used by the encoder. Assume that the induced input distributions satisfy
\[d^\star(\bar P_{\alpha,m}^{(n)})\le D_n\]
for all $m\in[M]$.
Fix $\eta>0$ with $D+\eta>0$ and set
$D'=D+\eta$.
    Let \(\hat P_{x^n}\) denote the empirical distribution of \(x^n\).

    For a deterministic feedback function \(f=(f_1,\ldots,f_n)\in\mathcal F_n\),
    define \[\mathcal E(f)\subseteq\mathcal X^n\times\mathcal Y^n\] as the set of all
    pairs \((x^n,y^n)\) satisfying
    \[x_t=f_t(y^{t-1}),\qquad t=1,\ldots,n,\]
    and
    \[\left|
    N(a,b|x^n,y^n)
    -
    n\hat P_{x^n}(a)\bar W(b|a)
    \right|
    \le
    n^{3/4}\bar W(b|a),\] 
    $(a,b)\in\mathcal X\times\mathcal Y$.
    
    For \(y^n\in\mathcal Y^n\), define the section
    \[
    \mathcal E(f;y^n)
    :=
    \{x^n\in\mathcal X^n:(x^n,y^n)\in \mathcal E(f)\}.
    \]     
    For \(\gamma\ge 0\), define
    \begin{equation}
    \begin{aligned}
    B_{\gamma,D'}(f)
    :=
    \Bigl\{
    y^n\in\mathcal Y^n:\;
    \exists\,x^n\in \mathcal E(f;y^n)
    \ \text{s.t.} 
    \hat P_{x^n}\in\mathcal P_{D'}
    \ \text{and}\
    I(\hat P_{x^n};\bar W)>\gamma
    \Bigr\}.
    \end{aligned}
    \end{equation}
    and let
    $B_{\gamma,D'}^{c}(f)
    :=
    \mathcal Y^n\setminus B_{\gamma,D'}(f)$.
    For the randomized feedback strategy \(Q_F(\cdot|\alpha,m)\), define
    \begin{equation}
    B_{\gamma,D'}\bigl(Q_F(\cdot|\alpha,m)\bigr)
    :=
    \bigcup_{f:Q_F(f|\alpha,m)>0} B_{\gamma,D'}(f),
    \end{equation}
    Then
    \begin{equation}
    \begin{aligned}
    &\left|
    \bigcup_{m=1}^M
    \bigl(\mathcal D_{\alpha,m}\cap B_{\gamma,D'}(Q_F(\cdot\mid\alpha,m))\bigr)
    \right| 
    \le
    2^{
n\left[
\max_{\substack{
P\in\mathcal P_{D'}:\\
I(P;\bar W)\ge\gamma
}}
H(P\bar W)
\right]
+o(n)
}.
    \end{aligned}
  \label{eq:union-bound-cardinality}
\end{equation}
Moreover, the following bound holds:
   \begin{equation}
\begin{aligned}
    &\frac1M\sum_{m=1}^M\sum_f Q_F(f|\alpha,m)
    \bar W^n\bigl(
    \mathcal D_{\alpha,m}\cap B_{\gamma,D+\eta}^c(f)
    \mid f\bigr)
    \le
\frac1M 2^{n\gamma+o(n)}
+\frac{D_n}{D'}+o(1).
\end{aligned}
\label{eq:trim-bound}
\end{equation}
\end{lemma}
    % \begin{equation}
    % \color{orange}
    % \begin{aligned}
    %     &P_M Q_F
    %     \bar W^n\bigl(\mathcal D_{\alpha,M}\cap B_{\gamma,D'}^c(f)\mid \alpha\bigr)
    %     \le \frac1M\mathrm{e}^{n\gamma+o(n)}+o(1).
    % \end{aligned}
    % \end{equation}

    \begin{proof}
    % The bound in \eqref{eq:trim-bound} follows directly from  of
    % Venkatesan and Anantharam~\cite{venkatesan2002identification}, .
    % It remains to prove \eqref{eq:union-bound-cardinality}. 
    This is an adaptation of~\cite[Lemma~4.2]{venkatesan2002identification} to add the distortion constraint.
    We first clarify the role of the randomized feedback strategy. The encoder
    uses a randomized strategy \(Q_F(\cdot|\alpha,m)\), i.e., it first selects a
    deterministic feedback function \(f\in\mathcal F_n\). Conditioned on this
    realization \(f\), the input sequence is causally generated from the output
    sequence through \(x_t=f_t(y^{t-1})\). %Hence, as in \cite[Lemma~4.2]{venkatesan2002identification}, the typicality condition is imposed on the joint pair \((x^n,y^n)\).
    In the counting argument we count
    output sequences \(y^n\), but only those for which there exists an input
    sequence \(x^n\in \mathcal E(f;y^n)\).
    
    If \(y^n\in B_{\gamma,D'}\bigl(Q_F(\cdot|\alpha,m)\bigr)\), then, by definition,
    there exist \(f\in\mathcal F_n\) with \(Q_F(f|\alpha,m)>0\) and
    \(x^n\in \mathcal E(f;y^n)\) such that
    $\hat P_{x^n}\in\mathcal P_{D'}$, 
    $I(\hat P_{x^n};\bar W)>\gamma$.
    Let \(Q=\hat P_{y^n}\). Since \((x^n,y^n)\in \mathcal E(f)\), for every
    \(b\in\mathcal Y\), we have
\[
    \left|Q(b)-(\hat P_{x^n}\bar W)(b)\right|
    % &=
    % \left|
    % \sum_{a\in\mathcal X}
    % \left(
    % \frac{1}{n}N(a,b|x^n,y^n)
    % -
    % \hat P_{x^n}(a)\bar W(b|a)
    % \right)
    % \right| \\
    \le
    \sum_{a\in\mathcal X} n^{-1/4}\bar W(b|a) 
    \le |\mathcal X|n^{-1/4}.
\]
    Therefore, by continuity of entropy,
    $H(Q)
    \le
    H(\hat P_{x^n}\bar W)+o(1)$.
    Since \(\hat P_{x^n}\in\mathcal P_{D'}\) and
    \(I(\hat P_{x^n};\bar W)>\gamma\), it follows that
    \[H(Q)
    \le
    \max_{\substack{
    P\in\mathcal P_{D'}:\\
    I(P;\bar W)\ge\gamma
    }}
    H(P\bar W)
    +o(1).\]
   Since the entropy bound is uniform over all $m$ and $f$,
every sequence \(y^n\) in
    $\bigcup_{m=1}^M
    \bigl(\mathcal D_{\alpha,m}\cap
    B_{\gamma,D'}(Q_F(\cdot\mid\alpha,m))\bigr)$
    has a type \(Q\) whose entropy is at most
    $\max_{\substack{
    P\in\mathcal P_{D'}:\\
    I(P;\bar W)\ge\gamma
    }}
    H(P\bar W)+o(1)$.
   Since the number of \(n\)-types on \(\mathcal Y\) is polynomial in \(n\), and
each type class has cardinality at most \(2^{nH(Q)}\), we obtain
\eqref{eq:union-bound-cardinality}.

    To prove \eqref{eq:trim-bound}, first omit the distortion
constraint from $B_{\gamma,D'}(f)$ and denote the resulting
set by $B_\gamma(f)$.
%Then

{
The complement \(B_{\gamma,D'}^c(f)\) can occur if the input-output pair is atypical, if the empirical input distribution violates the distortion threshold \(D'\), or if its empirical mutual information does not exceed \(\gamma\).

Hence,
}
\[
\begin{aligned}
&\bar W^n(\mathcal D_{\alpha,m}
\cap B_{\gamma,D'}^c(f)\mid f)
\le
\bar W^n(\mathcal D_{\alpha,m}
\cap B_\gamma^c(f)\mid f)
\Pr\{d^\star(\hat P_{X^n})>D'\mid f\}
+\Pr\{(X^n,Y^n)\notin\mathcal E(f)\mid f\}.
\end{aligned}
\]
Since the decoding sets $\mathcal D_{\alpha,m}$ are
pairwise disjoint for fixed $\alpha$, the first term,
averaged over $m$ and $f$, is bounded by
$\frac1M 2^{n\gamma+o(n)}+o(1)$
according to \cite[Lemma~4.2]{venkatesan2002identification}.
{Moreover,
\begin{align*}
&\frac1M\sum_{m=1}^M\sum_f Q_F(f|\alpha,m)\mathbb E\!\left[d^\star(\hat P_{X^n})\mid f,\alpha,m\right]
=\frac1M\sum_{m=1}^M d^\star(\bar P_{\alpha,m}^{(n)})
\le D_n.
\end{align*}
Hence, averaging Markov’s inequality over \(m\) and \(f\) gives
}
% \sout{For the second term, Markov's inequality gives}
\[
\frac1M\sum_m\sum_f Q_F(f\mid\alpha,m)
\Pr\{d^\star(\hat P_{X^n})>D'\mid f\}
\le \frac{D_n}{D'}.
\]
The atypicality probability is $o(1)$ uniformly
over all feedback strategies. Combining the
bounds proves \eqref{eq:trim-bound}.
    \end{proof}

\section{Proof of Theorem~\ref{thm:main}}

\subsection{Code Construction}

% We first consider the case $I(P^\star;\bar W)=0$ and
% $H(P^\star\bar W)>0$, where necessarily $R_{\mathrm T}=0$.
% Set $M_n=L_n=1$ and choose a codeword $x^n$ of type
% $P_n\to P^\star$. Since
% $H(\bar W|P^\star)=H(P^\star\bar W)$, the output
% sequence can be partitioned into
% $K_n=\lfloor 2^{n(H(P^\star\bar W)-\eta)}\rfloor$
% approximately equiprobable sets, with an exceptional
% set of vanishing probability.
% Applying Lemma~\ref{lem:array} with $S_n=K_n$
% and transmitting the corresponding identification tag
% over $t_0=o(n)$ additional channel uses yields every
% $R_{\mathrm{ID}}<H(P^\star\bar W)$.
% The tag can be transmitted reliably since $C(\bar W)>0$.
% As the second phase has vanishing relative length,
% the asymptotic distortion is at most $d^\star(P^\star)\le D$.

  Fix an arbitrary distribution \(P^\star\in\mathcal P_{D}\).
We first consider the case $I(P^\star;\bar W)=0$ and
$H(P^\star\bar W)>0$. Here, only $R_{\mathrm T}=0$
is possible. For any $R_{\mathrm{ID}}<H(P^\star\bar W)$,
choose \[0<\eta<H(P^\star\bar W)-R_{\mathrm{ID}}.\]
Set $M_n=L_n=1$ and choose a codeword
$x^n$ of type $P_n\to P^\star$.
Since \[H(\bar W|P^\star)=H(P^\star\bar W),\]
the conditional AEP for the fixed codeword $x^n$
implies that its typical output sequences have
probabilities at most
$2^{-n(H(P^\star\bar W)-\eta/2)}$
for all sufficiently large $n$.
Hence, they can be partitioned into
\[K_n=\lfloor 2^{n(H(P^\star\bar W)-\eta)}\rfloor\]
approximately equiprobable sets, with an exceptional
set of probability at most $\beta_n\to0$, and
total deviation from the uniform distribution
at most $\delta_n\to0$.
Fix $\varepsilon\in(0,1)$ and an integer $J\ge 2/\varepsilon$.
Choose $N_n=\lfloor 2^{2^{nR_{\mathrm{ID}}}}\rfloor$.
Since \[\log N_n\le c_\varepsilon K_n\] for all
sufficiently large $n$, Lemma~\ref{lem:array}
provides an identification array
\[A:[N_n]\times[K_n]\to[J].\]
The feedback reveals the partition index $k$ to the
encoder, which transmits the tag $A(\alpha,k)$
using a length-$t_0$ transmission code, where
$t_0=\lceil\sqrt n\rceil$.
Since $J$ is fixed and $C(\bar W)>0$, such a code
exists with maximal error probability
$\epsilon_n^{\mathrm{tag}}\to0$.
Letting $\eta\downarrow0$ yields every
$R_{\mathrm{ID}}<H(P^\star\bar W)$.
The first-kind error is at most
$\beta_n+\epsilon_n^{\mathrm{tag}}$, while the
second-kind error is at most
$\beta_n+\delta_n+\epsilon_n^{\mathrm{tag}}+\varepsilon$.
Thus, both errors vanish by subsequently letting
$\varepsilon\downarrow0$.
As the second phase has vanishing relative length,
the asymptotic distortion is at most $d^\star(P^\star)\le D$.

We now consider $I(P^\star;\bar W)>0$ and a rate pair
$(R_{\mathrm{ID}},R_{\mathrm T})$ satisfying
\begin{equation}
\label{eq:RIDRT}
R_{\mathrm{ID}}<H(P^\star\bar W),
\qquad
R_{\mathrm T}<I(P^\star;\bar W).
\end{equation}
Choose $\eta>0$ sufficiently small such that
\[
R_{\mathrm{ID}}<H(P^\star\bar W)-3\eta,
\qquad
R_{\mathrm T}<I(P^\star;\bar W)-3\eta,
\]
and, if $H(\bar W|P^\star)>0$, also
$\eta<H(\bar W|P^\star)$.
    Let \(P_n\) be a sequence of \(n\)-types such that \(P_n\to P^\star\), define
    \begin{equation}
   M_n:=\lfloor 2^{nR_{\mathrm T}}\rfloor,\qquad
L_n:=\left\lfloor
2^{n(I(P^\star;\bar W)-R_{\mathrm T}-2\eta)}
\right\rfloor
    \end{equation}
    By Lemma~\ref{lem:equitype_partition} when
$H(\bar W|P^\star)>0$, and by standard
constant-composition channel coding otherwise,
there exist, for all sufficiently large $n$, codewords
    \(c_{m,\ell}\in T_{P_n}\), where \(T_{P_n}\) denotes the type class in
    \(\mathcal X^n\) with empirical distribution \(P_n\), and pairwise disjoint
    decoding regions \(C_{m,\ell}\subseteq\mathcal Y^n\),
    \((m,\ell)\in[M_n]\times[L_n]\), satisfying
   \[\bar W^n(C_{m,\ell}^c|c_{m,\ell})\le\alpha_n,\] for $\alpha_n\to0$.
 If $H(\bar W|P^\star)>0$, Lemma~\ref{lem:equitype_partition}
provides a further partition of each decoding region,
\[C_{m,\ell}
=
C_{m,\ell,e}
\;\dot\cup\;
\left(
\dot\bigcup_{k=1}^{K_n} C_{m,\ell,k}
\right),\]
where we set
\[
K_n:=\left\lfloor
2^{n(H(\bar W|P^\star)-\eta)}
\right\rfloor.
\]
In this case, the sets $C_{m,\ell,k}$ are approximately
equiprobable under $\bar W^n(\cdot|c_{m,\ell})$, while
\[\bar W^n(C_{m,\ell,e}|c_{m,\ell})\le\beta_n\]
with $\beta_n\to0$.

If $H(\bar W|P^\star)=0$, we instead set
$K_n=1$, $C_{m,\ell,1}=C_{m,\ell}$,
and $C_{m,\ell,e}=\emptyset$.
    
    Define
    \[B_n:=M_nL_nK_n.\] Then
    % \[
    % \frac1n\log B_n
    % \le
    % R_{\mathrm T}
    % +
    % I(P^\star;\bar W)-R_{\mathrm T}-2\eta
    % +
    % H(\bar W|P^\star)-\eta .
    % \]
    % Hence
  \[
\frac1n\log B_n
=
\begin{cases}
H(P^\star\bar W)-3\eta+o(1),
& H(\bar W|P^\star)>0,\\
H(P^\star\bar W)-2\eta+o(1),
& H(\bar W|P^\star)=0.
\end{cases}
\]
    
  Fix $\varepsilon\in(0,1)$ and an integer $J\ge 2/\varepsilon$.
Let \(N_n:=\left\lfloor 2^{2^{nR_{\mathrm{ID}}}}\right\rfloor\).
Since \[R_{\mathrm{ID}}<H(P^\star\bar W)-3\eta,\]
we have \(\log N_n\le c_\varepsilon B_n\) for all sufficiently
large \(n\). Thus, Lemma~\ref{lem:array} yields an array
\begin{equation}
A:[N_n]\times[M_n]\times[L_n]\times[K_n]\to[J],
\end{equation}
such that any two distinct identification rows have Hamming distance
larger than $(1-\varepsilon)B_n$.
    
   Finally, set \(t_0=\lceil\sqrt n\rceil\).
Since \(I(P^\star;\bar W)>0\), there exists a length-\(t_0\)
transmission code
$\{(u_j,\widetilde D_j):j\in[J]\}$
for the averaged channel \(\bar W\) such that
\[\bar W^{t_0}(\widetilde D_j^c\mid u_j)
\le \epsilon_n^{\mathrm{tag}}\]
for all $j\in[J]$, where
$\epsilon_n^{\mathrm{tag}}\to0$.
    
    The total blocklength is \(n+t_0\).
    
    The two-phase encoder defined below induces a randomized feedback strategy
$Q_F(\cdot\mid\alpha,m)\in\mathcal P(\mathcal F_{n+t_0})$
for every message pair \((\alpha,m)\in[N_n]\times[M_n]$.
The randomization consists of the uniform choice of $\ell$.
    
    \paragraph{Phase I (transmission and output randomization)}
    
    To send a message pair
    $(\alpha,m),$
    the encoder chooses
    $\ell\in[L_n]$
    uniformly at random and transmits the first-phase codeword
    $c_{m,\ell}\in\mathcal X^n$.
    After Phase I, the encoder observes the noiseless feedback sequence \(y^n\).
    If
    $y^n\in C_{m,\ell,k}$
    for some \(k\in[K_n]\), the encoder records this index \(k\). If instead
    $y^n\in C_{m,\ell,e}$,
    an exceptional event is declared.
    
    \paragraph{Phase II (identification tag transmission)}
    
    If no exceptional event occurs in Phase I, the encoder computes
    $j=A(\alpha,m,\ell,k),$
    and transmits the second-phase tag codeword \(u_j\) over \(t_0\) channel uses.
    If the exceptional event occurs, the encoder transmits an arbitrary fixed
    fallback tag codeword.
    
    \paragraph{Decoding}
    
    Let the total channel output be written as
    $(y^n,\tilde y^{t_0})\in\mathcal Y^n\times\mathcal Y^{t_0}$.
    From the first-phase output \(y^n\), the receiver first finds the unique pair
    $(\hat m,\hat\ell)\in[M_n]\times[L_n]$
    such that
    $y^n\in C_{\hat m,\hat\ell}$.
  If $y^n\in C_{\hat m,\hat\ell,\hat k}$
for some \(\hat k\in[K_n]\), the receiver keeps the index \(\hat k\).
Otherwise, it declares a first-phase failure.
    
    From the second-phase output \(\tilde y^{t_0}\), the receiver finds the unique
    \(\hat j\in[J]\) such that
    $\tilde y^{t_0}\in\widetilde D_{\hat j}$.
    
    For every $(\alpha,m)$, define
    \begin{equation}
    \mathcal D_{\alpha,m}
    :=
    \bigcup_{\ell=1}^{L_n}
    \bigcup_{k=1}^{K_n}
    \left(
    C_{m,\ell,k}
    \times
    \widetilde D_{A(\alpha,m,\ell,k)}
    \right).
   \end{equation}
Thus, identity \(\alpha\) accepts and outputs transmission message \(m\)
    precisely when
    $(y^n,\tilde y^{t_0})\in\mathcal D_{\alpha,m}$.
    For fixed \(\alpha\), the sets \(\mathcal D_{\alpha,m}\) are pairwise disjoint
    in \(m\), because the first-phase decoding sets \(C_{m,\ell}\) are pairwise
    disjoint.
    
\subsection{Achievability}
    We show that the constructed code is an
\((n+t_0,N_n,M_n,\lambda_{1,n},\lambda_{2,n},D_n)\)
randomized ITS code satisfying
\[\lambda_{1,n},\lambda_{2,n}\to0,\qquad
\limsup_{n\to\infty}D_n\le D.\]
    \paragraph{Error probability of the first kind}
    Fix $(\alpha,m)$ and suppose that $(\alpha,m)$ is sent.
    An error of the first kind can occur only if
    $y^n\notin C_{m,\ell}$, or $y^n\in C_{m,\ell,e}$, or the second-phase tag is
    decoded incorrectly. Hence
\[\sum_f Q_F(f|\alpha,m)
\bar W^{n+t_0}(\mathcal D_{\alpha,m}^c|f)
\le \alpha_n+\beta_n+\epsilon_n^{\mathrm{tag}}.\]
    % The same bound holds for the original
    % state-dependent channel averaged over the i.i.d. states, i.e., \rf{is this needed? maybe just remove} \vg{text or equation? I kept a short statement here to make explicit that the bound for the averaged channel also applies to the original state-dependent channel after averaging over the i.i.d. state sequence, using \eqref{eq:state-averaging-identity}.} \rf{both, this next equation is literally substituting the definition, there is no conceptual difference between $\bar W$ and the averaged state-dependent channel, they are the same thing by definition, so this to me sounds like just repeating the same statement twice.}
    % \[
    % \begin{aligned}
    % &\sum_{s^{n+t_0}} P_S^{n+t_0}(s^{n+t_0})
    % \sum_f Q_F(f\mid \alpha,m)\,
    % W_S^{n+t_0}\!\left(\mathcal D_{\alpha,m}^{c}\mid f,s^{n+t_0}\right) \\
    % &\qquad\le \epsilon_n+\beta_n+\epsilon_1 .
    % \end{aligned}
    % \]
   Averaging over $m$ and taking the maximum over $\alpha$ gives
\[\lambda_{1,n}\le \alpha_n+\beta_n+\epsilon_n^{\mathrm{tag}}.\]
    \paragraph{Error probability of the second kind}
    
    Fix \(\alpha,\alpha'\in[N_n]\) with \(\alpha\neq\alpha'\), and suppose that
    \((\alpha,m)\) is sent. Receiver \(\alpha'\) can falsely accept only if either
    a first-phase decoding error occurs, a second-phase tag decoding error occurs,
    or the correct indices \((m,\ell,k)\) are recovered and nevertheless
    $A(\alpha,m,\ell,k)=A(\alpha',m,\ell,k)$.
    
  For fixed \(m\), let \(E_m(\alpha,\alpha')\) denote the number of pairs
$(\ell,k)\in[L_n]\times[K_n]$
such that
$A(\alpha,m,\ell,k)=A(\alpha',m,\ell,k)$.
By Lemma~\ref{lem:equitype_partition} when
$H(\bar W|P^\star)>0$, and by $K_n=1$ otherwise,
\[
\begin{aligned}
\frac{1}{M_n}\sum_{m=1}^{M_n}
\sum_f Q_F(f\mid\alpha,m)
\bar W^{n+t_0}(\mathcal D_{\alpha'}\mid f)
\le
\alpha_n+\beta_n+\epsilon_n^{\mathrm{tag}}
+\frac{1}{M_n}\sum_{m=1}^{M_n}
\frac{E_m(\alpha,\alpha')}{L_nK_n}
+\delta_n.
\end{aligned}
\]
By the Hamming-distance property of Lemma~\ref{lem:array},
\[
\frac{1}{M_n}\sum_{m=1}^{M_n}
\frac{E_m(\alpha,\alpha')}{L_nK_n}
\le\varepsilon.
\]
Consequently,
\[
\lambda_{2,n}\le
\alpha_n+\beta_n+\epsilon_n^{\mathrm{tag}}
+\delta_n+\varepsilon.
\]
    
    \paragraph{Sensing distortion}
    
    The sensing requirement is verified using the symbolwise estimator
    \begin{equation}
    \hat S_t=h^\star(X_t,Y_t), \qquad t=1,\ldots,n+t_0 .
    \end{equation}
    First consider Phase I. Since every first-phase codeword \(c_{m,\ell}\) has
    type \(P_n\), for every \[x\in\mathcal X,\qquad
    \frac1n\sum_{t=1}^n
    \mathbf 1\{c_{m,\ell,t}=x\}
    =
    P_n(x).\]
    
    Conditioned on the first-phase codeword \(c_{m,\ell}\), the average sensing
    distortion in Phase I is
    \[\frac1n\sum_{t=1}^n
    \mathbb E\bigl[d(S_t,\hat S_t)\mid X_t=c_{m,\ell,t}\bigr].\]
    By \eqref{eq:dstar}, this equals
    \begin{equation}
    \frac1n\sum_{t=1}^n d^\star(c_{m,\ell,t})
    =
    \sum_{x\in\mathcal X}P_n(x)d^\star(x)
    =
    d^\star(P_n).
    \end{equation}
    Since \(P_n\to P^\star\) and \(d^\star(\cdot)\) is linear in the input
    distribution, we have
    $d^\star(P_n)\to d^\star(P^\star)$.
 Since \(P^\star\in\mathcal P_D\), we have
$d^\star(P^\star)\le D$.
    
    It remains to account for Phase II. Since $t_0=o(n)$,
its contribution to the normalized block distortion vanishes.
Let
\begin{equation}
d_{\max}:=
\max_{(s,\hat s)\in\mathcal S\times\mathcal S}
d(s,\hat s)<\infty.
\end{equation}
Then the average distortion for every message pair over the complete block satisfies
\begin{equation}
\bar d^{(n+t_0)}
\le D_n:=
\frac{n}{n+t_0}d^\star(P_n)
+
\frac{t_0}{n+t_0}d_{\max}.
\end{equation}
Consequently,
\[
\limsup_{n\to\infty}D_n
=d^\star(P^\star)\le D.
\]
    %It remains to account for Phase II. Since Phase II has fixed length \(t_0\),
    %its contribution to the normalized block distortion vanishes as
    %\(n\to\infty\). More precisely, let
    %\[
    %d_{\max}:=
    %\max_{(s,\hat s)\in\mathcal S\times\hat{\mathcal S}}
    %d(s,\hat s)<\infty .
    %\]
    %For every transmitted message pair \((\alpha,m)\), the average distortion
    %over the complete block of length \(n+t_0\) satisfies
    %\[
    %\bar d_{\alpha,m}^{(n+t_0)}
    %\le
    %\frac{n}{n+t_0}d^\star(P_n)
    %+
    %\frac{t_0}{n+t_0}d_{\max}.
    %\]
    %Indeed, the first-phase distortion equals \(d^\star(P_n)\) for every
    %\((\alpha,m)\), since all first-phase codewords \(c_{m,\ell}\) have type
    %\(P_n\), while the second phase contributes at most \(d_{\max}\) per symbol.
    %Therefore,
    %\[
    %\bar d^{(n+t_0)}
    %=
    %\max_{(\alpha,m)}
    %\bar d_{\alpha,m}^{(n+t_0)}
    %\le
    %\frac{n}{n+t_0}d^\star(P_n)
    %+
    %\frac{t_0}{n+t_0}d_{\max}.
    %\]
    %Taking the limit superior as \(n\to\infty\), we obtain
    %\[
    %\limsup_{n\to\infty}\bar d^{(n+t_0)}
    %\le
    %d^\star(P^\star)
    %\le D.
    %\]

    For every fixed $\varepsilon>0$, the construction gives
\[\limsup_{n\to\infty}\lambda_{2,n}\le\varepsilon,\]
while $\lambda_{1,n}\to0$.
Choosing $\varepsilon=\varepsilon_n\downarrow0$ sufficiently slowly
such that $\log J_n=o(t_0)$, and applying a diagonal
argument yields a sequence of codes with
$\lambda_{1,n},\lambda_{2,n}\to0$.
    
    \paragraph{Rates.}
    
    By construction,
    \[\frac1{n+t_0}\log M_n\to R_{\mathrm T}.\]
    Moreover,
    $N_n=\left\lfloor 2^{2^{nR_{\mathrm{ID}}}}\right\rfloor$,
    and hence
    \[\frac1{n+t_0}\log\log N_n\to R_{\mathrm{ID}}.\]
    Thus, for $I(P^\star;\bar W)>0$, every rate pair satisfying
\eqref{eq:RIDRT} is achievable under distortion constraint $D$.
The case $I(P^\star;\bar W)=0$ and
$H(P^\star\bar W)>0$ was addressed separately above.
If $H(P^\star\bar W)=0$, both rates are zero and
achievability is trivial.

Since \(P^\star\in\mathcal P_{D}\) was arbitrary, the two
constructions establish achievability for all interior rate pairs
and the corresponding pure identification rates.
Approximating the boundary rates from below and applying
a diagonal argument yields the remaining rate pairs.
The zero-rate pair is trivially achievable.
Thus,
    \[\bigcup_{P\in\mathcal{P}_D}
    \left\{
    (R_{\mathrm{ID}},R_{\mathrm T}):
    R_{\mathrm{T}}\le I(P;\bar{W}),\;
    R_{\mathrm{ID}}\le H(P\bar{W})
    \right\}
    \subseteq
    \mathcal{C}_{\mathrm{ID,T}}^{\mathrm{r}}(D).\]

\subsection{Converse Proof}
    Let $(R_{\mathrm{ID}},R_{\mathrm T})$ be achievable under
distortion constraint $D$.
By Definition~2, there exists a sequence of
$(n,N_n,M_n,\lambda_{1,n},\lambda_{2,n},D_n)$
randomized ITS codes
$\left\{
\left(Q_F(\cdot|\alpha,m),\mathcal D_{\alpha,m}\right)
\right\}$
such that
$\lambda_{1,n},\lambda_{2,n}\to0$,
$\limsup_{n\to\infty}D_n\le D$, and
\[
\liminf_{n\to\infty}\frac1n\log M_n\ge R_{\mathrm T},
\qquad
\liminf_{n\to\infty}\frac1n\log\log N_n
\ge R_{\mathrm{ID}}.
\]

    For every $(\alpha,m)\in[N_n]\times[M_n]$, with \(\bar P_{\alpha,m}^{(n)}\) defined in Lemma~\ref{lem:dist_lb}, 
    \eqref{eq:Pbar}, %Lemma~\ref{lem:dist_lb}
    gives
  \[d^\star\bigl(\bar P_{\alpha,m}^{(n)}\bigr)
\le
\bar d_{\alpha,m}^{(n)}
\le
\bar d^{(n)}
\le D_n.\]
Thus,
$\bar P_{\alpha,m}^{(n)}\in\mathcal P_{D_n}$.

For each fixed $\alpha$, let
\[
\bar P_\alpha^{(n)}
:=
\frac1{M_n}\sum_{m=1}^{M_n}
\bar P_{\alpha,m}^{(n)}.
\]
By linearity of $d^\star$, we have
$d^\star(\bar P_\alpha^{(n)})\le D_n$.
Since the decoding sets are disjoint for fixed $\alpha$,
Fano's inequality and the feedback channel converse yield
% \[
% \begin{aligned}
% (1-\lambda_{1,n})\log M_n
% &\le I(M;Y^n\mid\alpha)+1\\
% &\le \sum_{t=1}^n I(X_t;Y_t\mid\alpha)+1\\
% &\le nI(\bar P_\alpha^{(n)};\bar W)+1.
% \end{aligned}
% \]
{
\[
\begin{aligned}
(1-\lambda_{1,n})\log M_n
\le I(M;Y^n|\alpha)+1
=\sum_{t=1}^n I(M;Y_t|Y^{t-1},\alpha)+1\\
\le \sum_{t=1}^n I(X_t;Y_t|Y^{t-1},\alpha)+1
\le \sum_{t=1}^n I(X_t;Y_t|\alpha)+1
\le nI(\bar P_\alpha^{(n)};\bar W)+1.
\end{aligned}
\]}

Define
\[
\bar P^{(n)}
:=
\frac1{N_n}\sum_{\alpha=1}^{N_n}
\bar P_\alpha^{(n)}.
\]
By linearity of $d^\star$ and concavity of mutual
information,
\[
\begin{aligned}
d^\star(\bar P^{(n)})&\le D_n,\qquad
(1-\lambda_{1,n})\frac{\log M_n}{n}
\le I(\bar P^{(n)};\bar W)+\frac1n.
\end{aligned}
\]

By compactness and continuity, these bounds imply
$R_{\mathrm T}\le\max_{P\in\mathcal P_D}I(P;\bar W)$.

We first consider $R_{\mathrm T}=0$.
Fix $\eta>0$ and set $D'=D+\eta$.
For each feedback strategy $f$, define
\[
B_{D'}^0(f):=
\{y^n:\exists x^n\in\mathcal E(f;y^n),
\ d^\star(\hat P_{x^n})\le D'\}.
\]
Define
\[
B_{D'}^0(Q_F):=
\bigcup_{f:Q_F(f)>0}B_{D'}^0(f)
\]
and, for each $\alpha\in[N_n]$,
\[
\mathcal D_\alpha^{0,\star}:=
\bigcup_{m=1}^{M_n}
\left(
\mathcal D_{\alpha,m}\cap
B_{D'}^0(Q_F(\cdot\mid\alpha,m))
\right).
\]
By conditional typicality and Markov's inequality,
these trimmed regions retain probability at least
\[
1-\lambda_{1,n}-\frac{D_n}{D'}-o(1).
\]
This is bounded away from zero for sufficiently
large $n$. Since $\lambda_{2,n}\to0$, these regions
are pairwise distinct.
Moreover, for every $y^n\in B_{D'}^0(f)$,
conditional typicality and continuity of entropy give
\[
H(\hat P_{y^n})
\le \max_{P\in\mathcal P_{D'}}H(P\bar W)+o(1).
\]
Since the number of output types is polynomial in $n$,
the type-counting argument of
Lemma~\ref{lem:dist_trim_count} yields
\[
R_{\mathrm{ID}}
\le \max_{P\in\mathcal P_{D'}}H(P\bar W).
\]
Letting $\eta\downarrow0$ gives
\[
R_{\mathrm{ID}}
\le \max_{P\in\mathcal P_D}H(P\bar W).
\]
We now consider $R_{\mathrm T}>0$.
    
 For $R_{\mathrm T}>0$, fix $0<\delta<R_{\mathrm T}$
and $\eta>0$, and set
\begin{equation}
\gamma:=R_{\mathrm T}-\delta,
\qquad D':=D+\eta.
\end{equation}
    For each \(\alpha\in[N_n]\), define the trimmed decoding region
    \begin{equation}
    \mathcal D_\alpha^\star
    :=
    \bigcup_{m=1}^{M_n}
    \left(
    \mathcal D_{\alpha,m}
    \cap
    B_{\gamma,D'}(Q_F(\cdot\mid\alpha,m))
    \right),
    \end{equation}
    where \(B_{\gamma,D'}(f)\) is defined in Lemma~\ref{lem:dist_trim_count}.
    
    % \vd{Converse gap: The average distortion constraint only implies
% E[d^\star(\hat P_{X^n})] <= D_n. It does not guarantee that
% \hat P_{X^n} belongs to \mathcal P_{D+\eta} with probability
% tending to one. The trimming estimate must therefore be
% justified for randomized strategies under an expected
% distortion constraint.}

By Lemma~\ref{lem:dist_trim_count},
\[
\begin{aligned}
&\frac1{M_n}\sum_{m=1}^{M_n}
\sum_f Q_F(f|\alpha,m)
\bar W^n\bigl(
\mathcal D_{\alpha,m}\cap B_{\gamma,D'}^c(f)
\mid f\bigr)
\le
\frac1{M_n}2^{n\gamma+o(n)}
+\frac{D_n}{D'}+o(1).
\end{aligned}
\]
   Since
\[
\liminf_{n\to\infty}\frac1n\log M_n
\ge R_{\mathrm T}
\]
and $\gamma=R_{\mathrm T}-\delta$, we have
\[
\frac1{M_n}2^{n\gamma+o(n)}\to0.
\]
    Indeed,
    {
    \begin{align*}
    &\frac1{M_n}\sum_{m=1}^{M_n}
    \sum_f Q_F(f|\alpha,m)
    \bar W^n(\mathcal D_\alpha^\star|f)
    \nonumber\\
    % &=
    % \frac1{M_n}\sum_{m=1}^{M_n}
    % \sum_f Q_F(f|\alpha,m)
    % \bar W^n\Bigl(\bigcup
    % \mathcal D_{\alpha,m}
    % \cap B_{\gamma,D'}(Q_F(\cdot|\alpha,m))
    % \mid f
    % \Bigr)
    % \nonumber\\
    &\ge
    \frac1{M_n}\sum_{m=1}^{M_n}
    \sum_f Q_F(f|\alpha,m)
    \bar W^n\Bigl(
    \mathcal D_{\alpha,m}
    \cap B_{\gamma,D'}(Q_F(\cdot|\alpha,m))
    \mid f
    \Bigr)
    \nonumber\\
    &=
    \frac1{M_n}\sum_{m=1}^{M_n}
    \sum_f Q_F(f|\alpha,m)
    \bar W^n(\mathcal D_{\alpha,m}|f)
    \nonumber\\
    &
    -
    \frac1{M_n}\sum_{m=1}^{M_n}
    \sum_f Q_F(f|\alpha,m)
    \bar W^n\Bigl(
    \mathcal D_{\alpha,m}
    \cap B_{\gamma,D'}^c(Q_F(\cdot|\alpha,m))
    \mid f
    \Bigr)
    \nonumber\\
   &\ge 1-\lambda_{1,n}
-\frac{D_n}{D'}-o(1).
    \end{align*}
    }
   where the last step follows from the first-kind error constraint,
Lemma~\ref{lem:dist_trim_count}, and the inclusion
\[B_{\gamma,D'}^c(Q_F(\cdot\mid\alpha,m))
\subseteq B_{\gamma,D'}^c(f)\]
for $Q_F(f\mid\alpha,m)>0$.
    
 For \(\alpha'\neq\alpha\), since
\(\mathcal D_{\alpha'}^\star\subseteq\mathcal D_{\alpha'}\), the second-kind
error constraint gives
\[
\frac1{M_n}\sum_{m=1}^{M_n}
\sum_f Q_F(f|\alpha,m)
\bar W^n(\mathcal D_{\alpha'}^\star|f)
\le \lambda_{2,n}.
\]
Since $\limsup_{n\to\infty}D_n\le D$ and
$\lambda_{1,n},\lambda_{2,n}\to0$, we have
\[
\liminf_{n\to\infty}
\left(
1-\lambda_{1,n}-\frac{D_n}{D'}-o(1)
\right)
\ge 1-\frac{D}{D'}
=\frac{\eta}{D+\eta}>0.
\]
Thus, for all sufficiently large $n$, the retained
probability exceeds $\lambda_{2,n}$.
Hence, the sets $\mathcal D_\alpha^\star$,
$\alpha\in[N_n]$, are pairwise distinct.
    
  By Lemma~\ref{lem:dist_trim_count}, each trimmed set
satisfies $|\mathcal D_\alpha^\star|\le K_n$, where
    \begin{equation}
   K_n
:=
2^{
n\left[
\max_{\substack{P\in\mathcal P_{D'}:\\ I(P;\bar W)\ge R_{\mathrm T}-\delta}}
H(P\bar W)
\right]
+o(n)
}.
    \end{equation}
   Since the \(N_n\) sets \(\mathcal D_\alpha^\star\) are distinct subsets of
\(\mathcal Y^n\) and satisfy \(|\mathcal D_\alpha^\star|\le K_n\), we have
\[
N_n\le
\sum_{j=0}^{\min\{\lfloor K_n\rfloor,|\mathcal Y|^n\}}
\binom{|\mathcal Y|^n}{j}.
\]
    Using
\[
\sum_{j=0}^{\min\{\lfloor K_n\rfloor,|\mathcal Y|^n\}}
\binom{|\mathcal Y|^n}{j}
\le
2^{K_n(n\log|\mathcal Y|+1)},
\]
   we obtain
\[
\frac1n\log\log N_n
\le
\max_{\substack{P\in\mathcal P_{D'}:\\
I(P;\bar W)\ge R_{\mathrm T}-\delta}}
H(P\bar W)+o(1).
\]
Combining this with the achievable identification rate gives
\[
R_{\mathrm{ID}}
\le
\max_{\substack{P\in\mathcal P_{D'}:\\
I(P;\bar W)\ge R_{\mathrm T}-\delta}}
H(P\bar W).
\]
Letting first $\delta\downarrow0$ and then
$\eta\downarrow0$, compactness and continuity yield
\[
R_{\mathrm{ID}}
\le
\max_{\substack{P\in\mathcal P_D:\\
I(P;\bar W)\ge R_{\mathrm T}}}
H(P\bar W).
\]
Since the feasible set is nonempty and compact and
$H(P\bar W)$ is continuous, the maximum is attained.
Thus, there exists \(P^\star\in\mathcal P_D\) such that
    $R_{\mathrm T}\le I(P^\star;\bar W)$, and
    $R_{\mathrm{ID}}\le H(P^\star\bar W)$.
    Thus,
\[
(R_{\mathrm{ID}},R_{\mathrm T})
\in
\bigcup_{P\in\mathcal P_{D}}
\left\{
(R_{\mathrm{ID}},R_{\mathrm T}):
\begin{aligned}
&0\le R_{\mathrm T}\le I(P;\bar W)
\\
&0\le R_{\mathrm{ID}}\le H(P\bar W)
\end{aligned}
\right\}.
\]
This completes the converse.

% \section{Conclusion}

%     We studied joint identification, transmission, and sensing over discrete memoryless state-dependent channels with noiseless feedback. For randomized encoding, we characterized the corresponding capacity-distortion region.
    
%     The result shows that, after averaging over the i.i.d. channel states, the communication part is governed by the averaged channel \(\bar W\), while the sensing constraint restricts the admissible input distributions through \(\mathcal P_{D'}\). Thus, sensing can reduce both the transmission and identification capabilities by limiting the set of feasible input laws.
    
%     Future work includes extensions to noisy feedback, causal or noncausal state information, secrecy constraints, and multiuser settings.

\section{Conclusion and Future Work}

    We studied joint identification, transmission, and sensing over discrete memoryless state-dependent channels with noiseless feedback and randomized encoding. The encoder sends both an identification message and a transmission message, while estimating the unknown state sequence under an average distortion constraint. For this model, we derived a single-letter characterization of the randomized identification-transmission-sensing capacity-distortion region.
    The characterization reveals a single-letter separation structure:
after averaging over the i.i.d. states, the communication rates
are determined by the averaged channel $\bar W$, subject to
the sensing-induced input constraint $P\in\mathcal P_D$.
For each feasible input distribution $P$, the transmission and
identification rates are bounded by $I(P;\bar W)$ and $H(P\bar W)$,
respectively. If $\mathcal P_D=\mathcal P(\mathcal X)$, the result
recovers the classical randomized identification-transmission region
with feedback. Otherwise, sensing may restrict the achievable rates
through the admissible input distributions.

    Future work includes noisy or rate-limited feedback, state information at the encoder, secrecy constraints for secure ISAC, multiuser models, Gaussian channels, and finite-blocklength analysis.

\section*{Acknowledgment}
    Christian Deppe and Vida Gholamian were supported in part by the German Federal Ministry of Research, Technology, and Space (BMFTR) through the Transfer Hub 6G-life under Grant 16KIS2415, and by the German Research Foundation (DFG) under Project DE1915/2-1. Christian Deppe, Roberto Ferrara, and Vida Gholamian were further supported in part by the BMFTR within the national initiative QSTARS under Grant 16KIS2602. Marcel Mross and Eduard Jorswieck were supported in part by the Federal Ministry of Research, Technology, and Space (BMFTR) through the xG-RIC project as part of the research program Communication Systems “Souverän. Digital. Vernetzt.” (grant number 16KIS2435).

\bibliographystyle{IEEEtran}
%\bibliographystyle{IEEEtranS}
%\balance
\bibliography{references}

\end{document}